%% file: ms.tex
\documentclass[letterpaper, 10 pt, conference]{ieeeconf}  

\IEEEoverridecommandlockouts                              
\usepackage{graphics} 
\usepackage{epsfig} 
\usepackage{mathptmx} 
\usepackage{times} 

\usepackage{amsthm}
\usepackage{amsmath}
\usepackage{amssymb}  
\usepackage{amsfonts}
\usepackage{algpseudocode}
\usepackage{algorithm}
\usepackage{xcolor}

\usepackage{cite}

\usepackage{graphicx}
\usepackage{float,epsf,caption,subcaption}

\newtheorem{problem}{Problem}
\newtheorem{theorem}{Theorem}

\theoremstyle{plain}
\newtheorem{lemma}{Lemma}
\theoremstyle{definition}
\newtheorem{definition}{Definition}
\newtheorem{example}{Example}
\theoremstyle{remark}

\title{\LARGE \bf
Safety Control with Preview Automaton
}

\author{Zexiang Liu and Necmiye Ozay
	\thanks{Zexiang Liu and Necmiye Ozay are with the Dept. of Electrical Engineering and Computer Science, Univ. of Michigan, Ann Arbor,
		MI 48109, USA
		{\tt\small zexiang,necmiye@umich.edu}. This work was supported by ONR grant N00014-18-1-2501, NSF grant ECCS-1553873, and an Early Career Faculty grant from NASA's Space Technology Research Grants Program. A poster abstract on our preliminary results was presented at HSCC \cite{zexiang2019poster}.}%
}

\begin{document}
	
	\maketitle
	\thispagestyle{empty}
	\pagestyle{empty}

	\begin{abstract}
	This paper considers the problem of safety controller synthesis for systems equipped with sensor modalities that can provide preview information. We consider switched systems where switching mode is an external signal for which preview information is available. In particular, it is assumed that the sensors can notify the controller about an upcoming mode switch before the switch occurs. We propose preview automaton, a mathematical construct that captures both the preview information and the possible constraints on switching signals. Then, we study safety control synthesis problem with preview information. An algorithm that computes the maximal invariant set in a given mode-dependent safe set is developed. These ideas are demonstrated on two case studies from autonomous driving domain.
		
	\end{abstract}

	\section{Introduction}
	\input{sec_intro.tex}

\noindent{\em Notation:} 
We use the convention that the set $ \mathbb{N} $ of natural number contains $ 0 $. Denote the expanded natural numbers as $ \overline{\mathbb{N}} = \mathbb{N}\cup \{\infty \} $. Denote the power set of set $ X $ as $ 2^{X} $.

\section{Problem Setup}\label{sec:setup}
\input{sec_problem.tex}

\section{Maximal Winning Set Computation with Preview Information }\label{sec:compute}

In this section, we propose an algorithm to solve Problem \ref{prob:1}. Recall that in Definition \ref{def:pa}, the feasible preview time interval given by $ T $ can be unbounded from the right, which is difficult to deal with in general because it essentially corresponds to a potentially unbounded clock. However, the following theorem reveals an important property of the preview automaton, which allows us to replace the preview time interval with its lower bound in the computation of winning sets.
\begin{theorem}
 Let $ G = \{ Q, E ,T,H\} $ and $ \widehat{G}= \{ Q, E ,\widehat{T},H\} $ be two preview automata of the switched system $ \Sigma $ with $ s $ modes, where $ \widehat{T}(q_1,q_2)  = \min (T(q_1,q_2)) $ for any $ (q_1,q_2)\in E $. Then given a safety specification $ \{S_i \}_{i\in Q} $, $ \{W_i\}_{i\in Q} $ is a winning set with respect to $ G $ if and only if $ \{W_i\}_{i\in Q} $ is a  winning set with respect to $ \widehat{G} $.\label{prop:equi}
\end{theorem}
\begin{proof}
Note that $ G $ and $ \widehat{G} $ are the same except the feasible preview time interval $ T $. For each transition $ (p_1,p_2)\in E $, $ \widehat{T}(q_1,q_2) $ is equal to the lower bound of $ T(q_1,q_2) $.
	
To show the ``only if" direction, suppose that $ W_i $ is a winning set with respect to $ G $ for mode $ i $. Then by Definition \ref{def:win_set}, there exists a controller $ \mathcal{U} $ such that any run of the closed-loop system with initial state in $ \{i\} \times W_i$ with respect to any valid preview input sequence of $G$ is safe. By Definition \ref{def:input_pa} and \ref{def:run_pa}, any preview input sequence and the corresponding execution of $ \widehat{G} $ are also valid preview inputs and execution of $ G $. Therefore, using the same controller $U$, any run of the closed-loop system 
with initial state in $ \{i\} \times W_i$ with respect to any valid preview input sequence of $\widehat{G} $ is safe.  Hence we conclude that each $ W_i $, for $ i\in Q $, is a winning set with respect to $ \widehat{G} $ for mode $i$.

To show the ``if" direction, suppose that $ W_i $ is a winning set with respect to $ \widehat{G} $ for mode $ i $, and $ \mathcal{U} $ is a controller such that any run of the closed-loop system  with initial state in $ \{i\} \times W_i$ with respect to any valid preview input sequence for $\widehat{G}$ is safe. Note that any execution of $G$ is also an execution of $ \widehat{G} $, but the corresponding preview input sequences of $G$ and $\widehat{G}$ can be different. Suppose that $\{( t_{k},\tau_{k},q_{k} )\} _{k=1}^{N}$ and $\{ (t_{k}', \tau_{k}', q_{k} )\}_{k=1}^{N}$ are two preview input sequences of $G$ and $\widehat{G}$ corresponding to the same execution $ \pi=\{(I_k,q_{k})\}_{k=0}^{N} $. Then by Definition \ref{def:input_pa}, for all $k\geq 1$, $ \tau'_{k}= \widehat{T}(q_{k},q_{k+1}) = \min(T(q_k,q_{k+1}))$ and $ \tau_{k}\in T(q_k,q_{k+1}) $ and $ t'_k+\tau'_k = t_k+\tau_k = \min(I_{k+1}) $. Hence $ \tau'_{k}\leq \tau_{k} $ and $ t'_k\geq t_k $ for all $ k\geq  1 $, which implies that for any transition in $\pi$, a controller of $ G $ always knows the next mode from the preview input of $G$ earlier than a controller of $ \widehat{G} $. 

Since the preview input of the next transition is earlier in $ G $ than in $ \widehat{G} $, given an execution $ \pi=\{(I_{k},q_{k})\}_{k=0}^{N} $, for any $ k\geq 1 $, $ \mathcal{U} $ can always infer\footnote{Given the $ k^{th} $ preview input $ (t_k,\tau_{k},q_{k+1}) $ of $ G $, the $ k^{th} $ preview input of $ \widehat{G} $ is $ (t_k+\tau_{k}-\tau'_{k},\tau'_{k},q_{k+1}) $ with $ \tau'_{k}=\widehat{T}(q_k,q_{k+1}) $.} 
the $ k^{th} $ preview input of $ \widehat{G} $ from the $ k^{th} $ input of $ G $ before time $ t'_k $. We force controller $ \mathcal{U} $ to generate control inputs for the switched system $ \Sigma $ with preview automaton $ G $ based on the inferred inputs of $ \widehat{G} $. Then any run of $ \Sigma $ when closing the loop with the customized $ \mathcal{U} $ and $ G $ is a run of the closed-loop system with respect to $ \Sigma $, $ \mathcal{U} $ and $ \widehat{G} $, which is safe if the initial state is in $ \{i\}\times W_i $. Hence $ W_i $ is a winning set of $ G $ for all $ i\in Q $.
\end{proof}

Thanks to Theorem \ref{prop:equi}, in terms of maximal winning set computation, it is enough to consider the preview automaton whose preview time interval is a singleton set for all transitions without introducing any conservatism. This property stated in Theorem \ref{prop:equi} can reduce computation cost and simplify the algorithms. Therefore, whenever there is a preview automaton $ G $, we first convert $ G $ into the form of $ \widehat{G} $ in Theorem \ref{prop:equi}. Algorithm \ref{alg:md2_nd} is designed to compute the maximal winning set for the preview automaton in the form of $ \widehat{G} $, whose result is equal to the maximal winning set of $ G $. {\color{black}}We note that $ \widehat{G} $ can be expanded to a non-deterministic finite transition system with $\sum_{i} (H_i -(\min_{j}T_{i,j})+\sum_{j}T_{i,j})$ states. Taking a product of this finite transition system with the switched system, the problem can be reduced to an invariance computation (with measurable and unmeasurable non-determinism) on the product system. However, the algorithms we propose avoid product construction and directly define fixed-point operations on the switched system's state space.

In Algorithm~\ref{alg:md2_nd}, lines $ 4 $-$ 6 $ compute the maximal winning set for each sink state in $ G $. Lines $7$-$12$ compute the winning sets of the non-sink states iteratively, with updates given by Algorithm~\ref{alg:invprend}. The main operators used in these algorithms are as follows. First, given a mode $ i $ of the switched system with state space $ X $ and action space $ U $, and a subset $ V $ of $ X $, the \emph{one-step controlled predecessor} of $ V $ with respect to the dynamics $ f_i $ is defined as
	\begin{equation}
		Pre^{f_i}(V) = \{x\in X: \exists u\in U, f_i(x,u)\subseteq V  \},
	\end{equation}
that is the set of states that can be guaranteed to reach the set $ V $ in one time step by some control inputs in $ U $.

Second, given a safe set $ S_i\subset X $, the \emph{one-step constrained controlled predecessors} $ PreInt(\cdot) $ of an arbitrary set $ V $ with respect to the dynamics $ f_i $ as 
	\begin{equation}
	PreInt^{f_i}(V,S_i) = Pre^{f_i}(V)\cap S_i.\label{eqn:cpre}
	\end{equation}
Now define $ V_0 = S_i $ and update $ V_k $ recursively for $ k\geq 0 $ by
	\begin{equation}
		V_{k+1} = PreInt^{f_i}(V_k,S_i). \label{eqn:fixed_point}
	\end{equation}
Note that $ \{V_k\}_{k=0}^{\infty} $ in \eqref{eqn:fixed_point} is monotonically non-increasing sequence of sets and the fixed point (reached when $ V_{k+1} = V_k $) is \emph{the maximal controlled invariant set} within the safe set $ S_i $ with respect to the dynamics $ f_i $, denoted as $ Inv^{f_i}(S_i) $.
Finally, given the preview automaton $ G= \{ Q,E,T,H\} $, the successors of some node $ i\in Q $ is defined as $ Post^G(i) = \{j:(i,j)\in E\} $.

\begin{algorithm}[]
	\caption{Winning Set for Problem 1}\label{alg:md2_nd}
	\begin{algorithmic}[1]
		\Function{$ ConInv $}{$ \Sigma,S = \{S_i\}_{i\in Q},G$}
		\State \textbf{initialize} $\{ W_i\}_{i\in Q}$ with $ W_i=S_i,\forall i\in Q $.\label{line:init}
		\State \textbf{initialize} $ \{V_i\}_{i\in Q} $ with $ V_i= \emptyset$.
		\For{$ i \in Q$ such that $ H(i)=\infty $ (sink states)}\label{line:for_sink}
		\State $ W_i \leftarrow Inv^{f_i}(S_i) $
		\EndFor\label{line:for_sink_end}
		\While{$\exists i\in Q $ such that  $ W_i \not= V_i$}\label{line:while}
		\State  $ V_i \leftarrow W_i, \forall i\in Q $
		\For{$ i \in Q$ such that $ H(i)<\infty $}\label{line:for_invpre}
		\State $ W_i \leftarrow InvPre^{f_i}(G,\{W_j\}_{j\in Post(i)},S)$\
		\EndFor\label{line:for_invpre_end}
		\EndWhile\label{line:while_end}
		\State \textbf{return} $ \{W_i\}_{i\in Q}$
		\EndFunction
	\end{algorithmic}
\end{algorithm}

\begin{algorithm}[]
	\caption{$ InvPre $ operator for Algorithm \ref{alg:md2_nd}}\label{alg:invprend}
	\begin{algorithmic}[1]
		\Function{$ InvPre^{f_i} $}{$ G, W, S $}
		\For{$ j $ in $ Post^{G}(i) $}\label{line:for_C_ij}
		\State $ C_{0,j} = W_j $ and $ T_{ij} = T(i,j) $
		\For{$ l $ = $ 1,2,3,...,T_{ij} $}
		\State $ C_{l,j} = PreInt^{f_i}(C_{l-1,j},S_{i}) $
		\EndFor
		\EndFor\label{line:end_for_C_ij}
		\State $ T_{min} = \min_{j\in Post^{G}(i)}{T(i,j)}$
		\State $ C_{T_{min}} = Inv^{f_i}(\bigcap_{j\in Post^{G}(i)} C_{T_{ij},j}) $
		\State $ H_i=H(i) $
		\For{$ k $ = $ T_{min}+1, \cdots, H_i $}
		\State $ C_{k} = PreInt^{f_i}(C_{k-1}, S_{i}) $\label{line:C_k_1}
		\If {$ J_k = \{j\in Post^{G}(i): T_{ij}\geq k\} \not=\emptyset$ }
		\State $ C_{k}  =  C_{k} \cap \big(\bigcap_{j\in J_k} C_{T_{ij},j}\big) $\label{line:C_k_2}
		\EndIf
		\EndFor
		\State \textbf{return} $ C_{H_i} $
		\EndFunction
	\end{algorithmic}
\end{algorithm}

Some properties of these operators and the $ InvPre $ operator defined by Algorithm~\ref{alg:invprend} are analyzed next. These properties are used later to prove the correctness of the main algorithm.
In what follows we use $\{\widehat{W}_i\}_{i\in Q} \subseteq \{W_i\}_{i\in Q} $ to denote the element-wise set inclusion $\widehat{W}_i\subseteq W_i $ for all $ i\in Q $.
When we talk about maximality, maximality is in (element-wise) set inclusion sense.

\begin{lemma}\label{lemma:1}
	Consider two collections of subsets $\widehat{W}= \{\widehat{W}_i\}_{i\in Q} $ and $W= \{W_i\}_{i\in Q} $ of $ X $. If $ \widehat{W}\subseteq W\subseteq S $,
	then $ \widehat{W}$ and $ W $ satisfy
	\begin{align}\label{eqn:mono_pre}
	InvPre^{f_i}(G,\widehat{W},S)\subseteq InvPre^{f_i}(G,W,S) \subseteq S_i
	\end{align}
	for any non-sink state $ i\in Q_{ns} $. 
\end{lemma}

\begin{lemma}\label{lemma:2}
	$ W=\{W_i\}_{i\in Q} $ is the maximal winning set with respect to the  safe set $ S=\{S_i\}_{i\in Q} $ if and only if $ \{W_i\}_{i\in Q_{ns}} $ is the 
	maximal solutions of the following equations:
	\begin{align}
	W_i=InvPre^{f_i}(G,W,S), \forall i\in Q_{ns}\label{eqn:invpre},
	\end{align}
	where the components of the winning set $W$ for sink states are chosen according to $ W_j = Inv^{f_j}(S_j) $ for all $ j\in Q_s $.
\end{lemma}
The proofs of Lemma 1 and 2 are given in the appendix. 

Let us illustrate how Algorithm~\ref{alg:md2_nd} works, using the switched system shown in Fig. \ref{fig:toydyn} with preview automaton in Fig. \ref{fig:toyexample} before proving that the proposed algorithm indeed computes the maximal winning set.

\begin{example}
Since there are no sink nodes in the preview automaton in Fig. \ref{fig:toyexample}, lines 4-6 in Algorithm \ref{alg:md2_nd} are skipped. We use pair $ (k,l) $ to indicate the $ k^{th} $ iteration of the while loop and $ l^{th} $ iteration of the for loop in line \ref{line:while} and \ref{line:for_invpre} in Algorithm \ref{alg:md2_nd}, and use $ W^{k,l}_i $ to refer to the value of $ W_i $ after the $ (k,l) $ iteration. Note that at iteration $ (k,l) $, only $ W_{l} $ is being updated and the other $ W_i $ remains unchanged for $ i\not= l $. 
	
Initially $ W^{0,0}_1 = W^{0,0}_2 = \{s_1,s_2\} $. In the iteration $ (0,1) $, $ W_1^{0,1} = InvPre^{f_1}(G,\{W^{0,0}_1,W^{0,0}_2\}, S)= \{s_1\}$ and $ W_2^{0,1} = W_2^{0,0} $. In the iteration $ (0,2) $, $ W_2^{0,2} = InvPre^{f_1}(G,\{W^{0,1}_1,W^{0,1}_2\}, S)= \{s_2\}$ and $ W_1^{0,2} = W_1^{0,1}$. In the following iterations $ (1,1),(1,2) $, $ W^{1,1} $ and $ W^{1,2} $ are unchanged. Therefore, the termination condition in line \ref{line:while} is satisfied and the output of Algorithm \ref{alg:md2_nd} of this example is $ W_1 = \{s_1\} $ and $ W_2 = \{s_2\} $. It is easy to verify that $ W_1 = \{s_1\} $ and $ W_2 = \{s_2\} $ form the maximal winning set for this problem. \qed 
\end{example}

\begin{figure}
	\centering
	\includegraphics[width=0.5\linewidth]{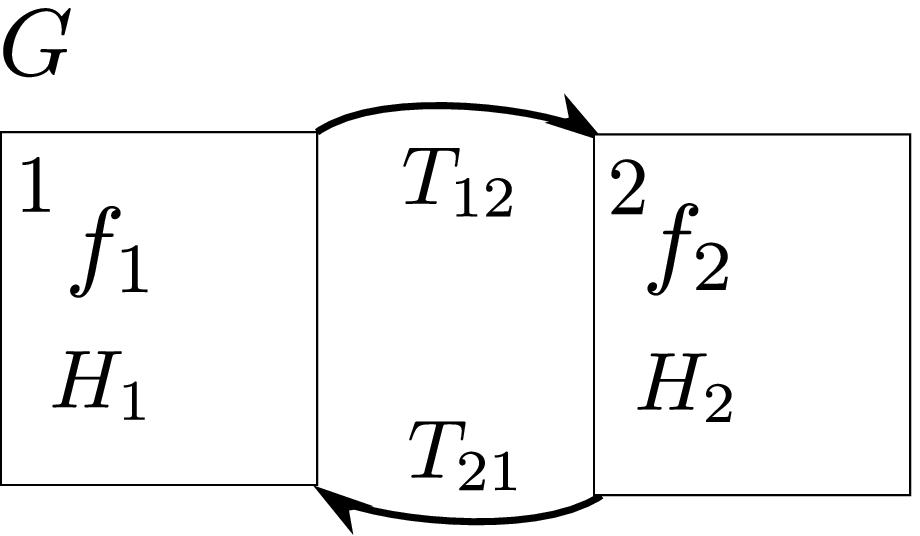}
	\caption{The preview automaton corresponding to the switched system in Fig. \ref{fig:toydyn}. $ H_1 = H_2 = 3 $ is the least holding time for both modes, and $ T_{12}=T_{21}=1$ is the preview time for transitions $ (1,2) $ and $ (2,1) $.}
	\label{fig:toyexample}
\end{figure}

The main completeness result is provided next.

\begin{theorem}
	If Algorithm~\ref{alg:md2_nd} terminates, the tuple of sets $ \{W_{i}\}_{i=1}^{n} $ it returns is the maximal winning set within the safe set $S = \{S_i\}_{i\in Q}$ of the switched system $ \Sigma $ with the preview automaton $ G $.\label{prop:max}
\end{theorem}
\begin{proof}
	Suppose that $ \{W^*_i\}_{i\in Q} $ is the maximal winning set we are looking for. Let us partition the discrete state space $ Q $ into the set of sink states $ Q_s = \{q\in Q:H(q)=\infty \} $ and the set of non-sink states $ Q_{ns} = Q\backslash Q_s = \{q\in Q: H(q)< \infty \}$.
	
	If $ q $ is a sink state, once the system enters the mode $ q $, the system remains in mode $ q $ without any future switching. Therefore, the maximal winning set $ W_q $ is trivially the maximal controlled invariant set within the safe set $ S_q $ with respect to the dynamics of mode $ q $, that is $ W^*_q=Inv^{f_q}(S_q) $. In line $ 4$-$ 6 $ of Algorithm \ref{alg:md2_nd}, we compute the maximal winning sets for all the sink states.
	
	We have solved $ W^*_i $ for sink state $ i\in Q_s $.  Let us consider the maximal winning sets for non-sink states. We want to show that $ W $ being updated based on lines $ 7 $-$ 12 $ of Algorithm \ref{alg:md2_nd} converges to $ W^* $. Without loss of generality, assume that $ Q_{ns} = \{1,2,...,s_{ns}\} $ and let the ``for" loop in line \ref{line:for_invpre} iterate over the indices $ 1,2,...,s_{ns} $ in the natural order. 
	
	We use $ W^{k,l} $ to indicate the updated value of $ W $ after the $ k^{th} $ iteration of the ``while" loop (line \ref{line:while}) and the $ l^{th} $ iteration of the ``for" loop (line \ref{line:for_invpre}). Then the initial value of $ W $ is $ W^{0,0} = \{W^{0,0}_i\}_{i\in Q}$ where $ W^{0,0}_i = S_i $ for all $ i\in Q_{ns} $ and $ W^{0,0}_j=W^*_j $ for all $ j\in Q_s $. According to line $ 9 $-$ 10 $, for all $ k\geq 0$ and $ 0\leq l\leq s_{ns}-1 $, $ W^{k,l+1}_{i} = InvPre^{f_{i}}(G,W^{k,l},S) $ for $ i= l+1 $ and $ W^{k,l+1}_j = W^{k,l}_j$ for all $ j\not= l+1 $ and $ W^{k+1,0} = W^{k,s_{ns} } $.
	
	Now we want to prove that if $ W^*\subseteq W^{k,0} \subseteq S$ for some $ k\geq 0 $, then $ W^*\subseteq W^{k,l}\subseteq S $ for any $ l\in\{1,2,...,s_{ns}\} $ by induction. (Base case 1) Since we have $ W^*\subseteq W^{k,0}\subseteq S $, by Lemma 1 and  2, we have $$ W^*_i = InvPre^{f_i}(G,W^*,S)\subseteq InvPre^{f_i}(G,W^{k,0},S) = W^{k,1}_i\subseteq S_i $$ for $ i=1 $. Since $ W^*\subseteq W^{k,0} $ and $ W^{k,1}_j = W^{k,0}_j $ for all $ j\not= 1 $, we have $ W^* \subseteq W^{k,1}\subseteq S $.  (Induction hypothesis 1) Suppose that $ W^* \subseteq W^{k,l}\subseteq S $ for some $ 0\leq l\leq s_{ns}-1 $. Again, by Lemma 1 and 2, $ W^*\subseteq W^{k,l+1}\subseteq S $. Finally by induction, if $ W^*\subseteq W^{k,0} \subseteq S$, $ W^* \subseteq W^{k,l}\subseteq S $ for all $ l\in\{1,2,...,s_{ns}\} $.
	
	Then next we want to prove by induction that $ W^*\subseteq W^{k,0}\subseteq S $ for all $ k\geq 0 $. (Base case 2) $ W^*\subseteq W^{0,0}\subseteq S$ by construction. (Induction hypothesis 2) Suppose $ W^* \subseteq W^{k,0}\subseteq S $ for some $ k\geq 0 $. Then, we have proven that $ W^* \subseteq W^{k,s_{ns}}=W^{k+1,0} \subseteq S $. Therefore by induction, $ W^*\subseteq W^{k,0}\subseteq S $ for any $ k\geq 0 $. 
	
	The two induction arguments above prove that $ W^*\subseteq W^{k,l} $ for any $ k\geq 0 $ and $ 0\leq l\leq s_{ns} $.
	
	Now let us show that $ W^{0,0}, W^{0,1}, W^{0,2},..., W^{k,0},W^{k+1,1},...$ is a non-expanding sequence. Since $ W^{k,s_{ns}} = W^{k+1,0} $ for all $ k\geq 0 $, it suffices to show $ W^{k,l+1}\subseteq W^{k,l} $ for any $ k\geq 0 $ and $ 0\leq l\leq s_{ns}-1 $.
	
	(Base case 3) Note that $ InvPre^{f_i}(G,V,S)\subseteq S_i $ for arbitrary $ V\subseteq X $ and $ i\in Q $. Thus by definition $ W^{0,1}_i=InvPre^{f_i}(G,W^{0,0},S) \subseteq S_i = W^{0,0}_i$ for $ i= 1 $. Note that $ W^{0,1}_j = W^{0,0}_j $ for all $ j\not= 1 $. Thus $ W^{0,1}\subseteq W^{0,0} $. Now consider $ W^{0,2} $. Note that $ W^{0,2}_j=W^{0,1}_j $ for all $ j\not=2 $. For $ i= 2 $, $ W^{0,2}_i = InvPre^{f_i}(G,W^{0,1},S)\subseteq S_i = W^{0,1}_i $. Thus $ W^{0,2}\subseteq W^{0,1} $. Similarly, we have $ W^{0,s_{ns}}\subseteq...\subseteq  W^{0,1}\subseteq W^{0,0} $. 
	
	(Induction hypothesis 3) Suppose $W^{k,s_{ns}}\subseteq ...\subseteq W^{k,1} \subseteq W^{k,0}$ for some $ k > 0$. To show that $ W^{k+1,l+1}\subseteq W^{k+1,l}$ for all $ l $, we need another induction argument. (Base case 4) We know $ W^{k+1,0} = W^{k,s_{ns}} $, and $ W^{k+1,1}_j =  W^{k+1,0}_j$ for $ j\not= 1 $. For $ i= 1 $, $W^{k+1,0}_i= W^{k,s_{ns}}_i =  W^{k,1}_i = InvPre^{f_i}(G,W^{k,0},S)$. By induction hypothesis 3, $W^{k,s_{ns}}\subseteq W^{k,0} $ and thus by Lemma \ref{lemma:1} and \ref{lemma:2}, $$ W^{k+1,1}_i = InvPre^{f_i}(G,W^{k+1,0},S) \subseteq InvPre^{f_i}(G,W^{k,0},S) = W^{k+1,0}_i $$ for $ i=1 $, and therefore $ W^{k+1,1}\subseteq W^{k+1,0} $.
	
	(Induction hypothesis 4) Suppose that $ W^{k+1,l}\subseteq W^{k+1,l-1}\subseteq ...\subseteq W^{k+1,0} $. By definition, $ W^{k+1,l+1}_j =  W^{k+1,l}_j$ for all $ j\not= l+1 $. Also, for $ i= l+1 $, $ W^{k+1,l}_i =   W^{k,l+1}_i = InvPre^{f_i}(G,W^{k,l},S)$. By the induction hypothesis 3 and 4, $ W^{k+1,l}\subseteq W^{k,l} $ and thus by Lemma \ref{lemma:1} and \ref{lemma:2} again, for $ i=l+1 $, $$ W^{k+1,l+1}_i = InvPre^{f_i}(G,W^{k+1,l},S) \subseteq InvPre^{f_i}(G,W^{k,l},S) =W^{k+1,l}_i $$ and therefore $ W^{k,l+1}\subseteq W^{k,l} $. Then by induction 4, we have $ W^{k+1,s_{ns}} \subseteq ... \subseteq W^{k+1,1}\subseteq W^{k+1,0} $.
	
	Therefore by the induction 3, we show that $ W^{0,0},..., W^{k,0},W^{k+1,1},...$ is non-expanding.
	
	By far, we have shown that $ W^{0,0}W^{0,1} ... W^{k,0} W^{k,1}...$ is a monotonic non-expanding sequence within $ S $, which implies that the limit of this sequence $ W^{\infty,0} $ (the output of Algorithm \ref{alg:md2_nd}) exists and is contained by $ S $ thus safe. By line $ 7 $-$ 12 $ of Algorithm \ref{alg:md2_nd}, $ W^{\infty,0} $ is a solution of equations in \eqref{eqn:invpre}. Also, since for any $ k$ and $l $, $ W^*\subseteq W^{k,l} $ and $ W^* $ is the maximal solution of equations in \eqref{eqn:invpre}, we have $ W^{\infty,0}=W^{*} $. 	
\end{proof}

Note that the above proof also guarantees termination if the switched system under consideration has finitely many states. For switched systems with continuous state spaces, the non-expanding property of the computed sets guarantees convergence but termination in finite number of steps is not guaranteed, in general. For linear switched systems, termination can still be guaranteed using algorithms from \cite{de2004computation,rungger2017computing} by slightly sacrificing maximality (see also \cite{smith}).

Once the maximal winning set (or a winning set) $ W^* = \{W^*_i\}_{i\in Q} $ is obtained, a controller can be extracted roughly as follows: for a sink node $ i\in Q_{s} $, the allowable control inputs for each state in the controlled invariant set $ W^*_i $ can be obtained by applying the $ Pre $ operator to $ W^*_i $. For a non-sink node $ j\in Q_{ns} $, we need a ``invariance" controller to make sure the system state remain in $ W^*_i $ before a preview happens, and a ``reachability" controller for each transition $ (j,k)\in E $ and each possible preview time $ \tau_{jk}\in T(j,k) $ such that from the time point a preview is received by the controller, system state can guarantee to reach $ W_k $ in $ \tau_{jk} $ steps, where the allowable control input for each step can be obtained by applying the $ PreInt $ recursively for $ \tau_{jk} $ times. For the ``invariance" controller, we also need to make sure that the system state reaches certain parts of the maximal winning set based on the holding time (time steps elapsed since last transition) such that once a preview occurs, the system state is within the domain of the corresponding ``reachability" controller. The process of computing the ``reachability" controllers actually corresponds to line 2-7 in Algorithm \ref{alg:invprend}, and the process of computing the ``invariance" controller corresponds to line 9 and 12-14, if the preview automaton has a singleton preview time interval. The process can be generalized to general preview time intervals from the Algorithm \ref{alg:invprend} based on the description above.

\section{Case Studies}\label{sec:cases}
\input{sec_example.tex}

\section{Conclusions and Future Work}\label{sec:conclusions}
In this paper, we introduced preview automaton and provided an algorithm for safety control synthesis in the existence of preview information. 
The proposed algorithm is shown to compute the maximal winning set upon termination. These ideas are demonstrated
with two examples from the autonomous driving domain. As shown in these examples, incorporation of preview information
in control synthesis leads to less conservative safety guarantees compared to standard controlled invariant set based approaches.
 In the future, we will investigate the 
use of preview automaton for synthesizing controllers from more general specifications. We also have some ongoing work
investigating the connections of preview automaton with discrete-time I/O hybrid automaton with clock variables representing preview and holding times.

\bibliographystyle{IEEEtran}
\bibliography{cdc}

\section*{Appendix}
\input{sec_app}
\end{document}

%% file: sec_intro.tex

Modern autonomous systems, like self-driving cars, unmanned aerial vehicles, or robots, are equipped with sensors like cameras, 
radars, or GPS that can provide information about what lies ahead. 
Incorporating such preview information in control and decision-making is an appealing idea to improve system performance. For instance, a considerable amount of work has been done for designing closed-form optimal control strategies with limited preview on future reference signal \cite{sheridan1966three,tomizuka1975optimal,katayama1985design,hazell2008discrete} with applications in vehicle control \cite{peng1993preview,xu2019design}. Preview information or forecasts of external factors can also be easily incorporated in a model predictive control framework \cite{garcia1989model,laks2011model}.

However, similar ideas are less explored in the context of improving system's safety assurances. For the aforementioned methods on optimal control with preview information, extra constraints need to be introduced into the optimal control problem to have safety guarantees, for which the closed-form optimal solution cannot be derived easily in general. In model predictive control framework, the online optimization is only tractable over a finite receding horizon, and thus it is difficult to have assurance on safety for which the constraints need to be satisfied over an infinite horizon. The goal of this paper is to develop a framework to enable incorporation of preview information in correct-by-construction control.

Specifically, we consider discrete-time switched systems where the mode signal is controlled by external factors. We assume the system is 
equipped with sensors that provide preview information
on the mode signal (i.e., some future values of the mode signal can be sensed/predicted at run-time). To capture
how the mode signal evolves and how it is sensed/predicted at run-time, we introduce \emph{preview automaton}.
Then, we focus on safety specifications defined
in terms of a safe set within each mode, and develop an algorithm that computes the 
maximal invariant set inside these safe sets while incorporating the preview information. A simple example where such information
can be relevant is depicted in Fig.~\ref{fig:example}, where an autonomous vehicle can use its forward looking sensors or GPS and map information
to predict when the road grade will change. The proposed framework provides a means to leverage such information to compute provably-safe controllers
that are less conservative compared to their preview agnostic counterparts. 

\begin{figure}
	\centering
	\includegraphics[width=1\linewidth]{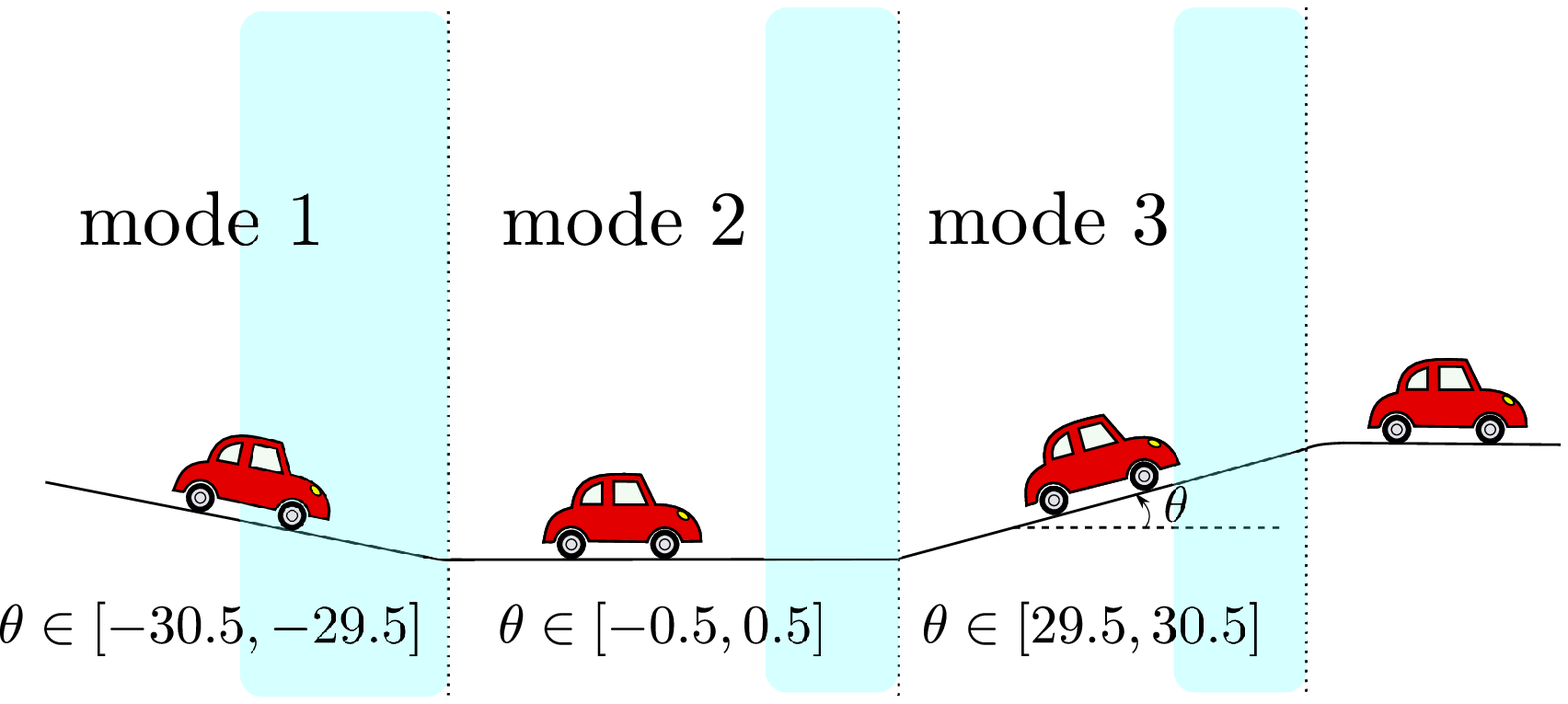}
	\caption{A simple example on autonomous vehicle cruise control. The road grade alternates between three ranges $r_1$, $r_2$, $r_3 $,  modeled by a switched system $ \Sigma$ with three modes. The blue shadows indicate the regions where the vehicle's sensors are able to look ahead and upon the detection of the upcoming change, a preview input is released.}
	\label{fig:example}
\end{figure}

Our work is related to \cite{kupferman2011synthesis,holtmann2010degrees,zimmermann2017finite} where synthesis from linear temporal logic or general omega-regular 
specifications are considered for discrete-state systems. In \cite{kupferman2011synthesis}, it is assumed that a fixed horizon lookahead is available; 
whereas, in our work, 
the preview or lookahead time is non-deterministic, and preview automaton can be composed with both discrete-state and continuous-state
systems. While we restrict our attention to safety control synthesis, extensions to other logic specifications are also possible.
Another main difference with \cite{kupferman2011synthesis} is that we use the preview automaton also to capture constraints 
on mode switching. The idea of using automata or temporal logics to capture assumptions on mode switching is used in \cite{athanasopoulos2016safety}
and \cite{nilsson2012temporal}. In particular, the structure of the controlled invariant sets we compute is similar to the invariant sets (for systems without control)
in \cite{athanasopoulos2016safety}. However, neither \cite{athanasopoulos2016safety} nor \cite{nilsson2012temporal} takes into account preview information.

The remainder of this paper is organized as follows. After briefly introducing the basic notations next, in Section \ref{sec:setup}, we describe the problem setup,
 define the preview automaton and formally
state the safety control problem. An algorithm to solve
the safety control problem with preview is proposed and analyzed in Section \ref{sec:compute}. In Section \ref{sec:cases}, we demonstrate the proposed 
algorithm by two case studies one on vehicle cruise control and another on lane keeping before we conclude the paper in Section \ref{sec:conclusions}.

%% file: sec_problem.tex

In this paper, we consider switched systems $\Sigma$ of the form:
\begin{align}
x(t+1) \in f_{\sigma(t)}(x(t), u(t)),\label{eqn:model}
\end{align}
where $\sigma(t)\in\{1, \ldots, s\}$ is the mode of the system, $x(t) \in X$ is the state and $u(t) \in U$ is the control input. We assume that the switching is uncontrolled (i.e., the mode $\sigma(t)$ is determined by the external environment) however $\sigma(t)$ is known when choosing $u(t)$ at time $t$. By defining each $f_i: X\times U \rightarrow 2^{X}$ to be set-valued, we capture potential disturbances and uncertainties in the system dynamics that are not directly measured at run-time but that affect the system's evolution. 

We are particularly interested in scenarios where some \emph{preview} information about the mode signal is available at run-time. That is, the system has the ability to lookahead and get notified of the value of mode signal before the mode signal switches value. More specifically, we assume that for each pair of modes $ (i,j) $, if the switching from $ i $ to $ j $ takes place next, a sensor can detect this switching for $ \tau_{ij} $ time steps ahead of the switching time, where $ \tau_{ij}\geq 0 $ is called \emph{the preview time} and belongs to a time interval $ T_{ij} $. Mathematically, if $\sigma(t+\tau_{ij}-1) =i $ and  $\sigma(t+\tau_{ij})=j$, then the value $\sigma(t+\tau_{ij})$ is available before choosing $u(t)$ at time $t$, for some $\tau_{ij}\in T_{ij}$.

In many applications, switching is not arbitrarily fast. That is, there is a minimal holding time (or, dwell time) between two consecutive switches. For each mode $ i $ of the switched system, we associate a \emph{least holding time} $ H_i\geq 1 $ such that if the system switches to mode $ i $ at time $ t $, the environment cannot switch to another mode at any time between $ t $ and $ t+H_i-1 $. Note that $ H_i=1 $ for all $ i $ is the trivial case where the system does not have any constraints on the least holding time. Moreover, there could be constraints on what modes can switch to what other modes. 

Following example illustrates some of the concepts above.

\begin{example}\label{ex:car_grade}
	In Figure \ref{fig:example}, a vehicle runs on a highway where the road grade can switch between ranges $ r_1=[-30.5,-29.5] $ and $ r_2=[-0.5,0.5] $ and between $ r_2 $ and $ r_3 = [29.5,30.5] $ (no direct switching between $ r_1 $ and $ r_3 $). We use a switched system $ \Sigma $ of $ 3 $ modes to model the three ranges $ r_1 $, $ r_2 $ and $ r_3 $. Thanks to the perception system on the vehicle, the switching from mode $ i $ to mode $ j $ can be detected $\tau_{ij}\in T_{ij} $ steps ahead, where $ T_{ij} $ is a known interval of feasible preview times for $ (i,j)\in \{(1,2),(2,1),(2,3),(3,2)\} $. Also, the least time steps for the vehicle in range $ r_i $ is $ H_i$ for $ i=1,2,3 $. \qed\label{eg:1}
\end{example}

For simplicity, in the rest of the paper, we will assume that the least holding time is greater than or equal to the least feasible preview time among all modes that the system can switched to from mode $ i $, i.e., $ H_i \geq \min(\cup_{j}T_{ij}) $ for any mode $ i $. This assumption is justified in many applications where switching is ``slow" compared to the worst-case sensor range. For instance, the road curvature or road grade does not change too frequently.

The main contribution of this paper is two folds:
\begin{itemize}
  \item to provide a new modeling mechanism for switched systems that can capture both the constraints on the switching and the preview information, 
 \item to develop algorithms that can compute
   controllers to guarantee safety with preview information in a way that is less conservative compared to their preview agnostic counterparts.
\end{itemize} 


\subsection{Preview Automaton}\label{sec:pa_def}
Provided the prior knowledge on the preview time interval $ T_{ij} $ and the least holding time $ H_i $, we model the allowable switching sequences of a switched system with preview with a mathematical construct we call \emph{preview automaton}.
\begin{definition}(\textbf{Preview Automaton})
	A \emph{preview automaton} $ G $ corresponding to a switched system $ \Sigma$ with $ s $ modes is a tuple $ G = \{ Q, E ,T,H\} $, where
	\begin{itemize}
		\setlength\itemsep{-0.8em}
		\item $ Q = \{1,2,\cdots, s\} $ is a set of nodes (discrete states), where node $ q\in Q $ corresponds to the mode $ q $ in $ \Sigma $;\\
		\item $ E\subseteq Q\times Q $ is a set of transitions;\\
		\item $ T:E\rightarrow \{[t_1,t_2]: 0\leq t_1\leq t_2, t_1\in \mathbb{N}, t_2\in \overline{\mathbb{N}} \} $ labels each transition with the time interval of possible preview times corresponding to that transition;\\
		\item  $ H:Q\rightarrow (\mathbb{N}\backslash \{0\})\cup \{\infty\} $  labels each node $ q \in Q $ with the least holding time corresponding to that node.
	\end{itemize}
	\label{def:pa}
\end{definition}
We make a few remarks. First, we do not allow any self-loops, i.e., $ (q,q) \not\in E $ for all $ q\in Q $. Second, the preview times $ T(q_1,q_2)$ for any $ (q_1,q_2)\in E $ is in one of the three forms: a singleton set $ \{t_1\}$ (interval $ [t_1,t_1] $), or a finite interval $ [t_1,t_2]$ with $ t_1<t_2 $, or an infinite interval $ [t_1,\infty)$. Finally, if there is a state $ q $ with no outgoing edges, that is $ \{(q,p)\in E:p\in Q \}=\emptyset$, we set $ H(q) = \infty $ to indicate that once the system visits $ q $, it remains in $ q $ indefinitely, so the deadlocks are not allowed. We call such a state a sink state. The set of sink states and the set of non-sink states in $ Q $ are denoted by $ Q_{s} $ and $ Q_{ns}$.

 In Definition \ref{def:pa}, the nodes of the preview automaton are chosen to be the modes of the switched system for simplicity. It is easy to extend the definition to allow multiple nodes in the preview automaton to correspond to the same mode. Alternatively, redefining the switched system by replicating certain modes and keeping the current definition can serve the same purpose.

\begin{example}
	The preview automaton for the switched system in Example \ref{eg:1} has nodes $  Q = \{1,2,3\} $ with transitions shown in Figure \ref{fig:previewautomata}. The least holding mapping $ H(q) = H_q $ for all $ q\in Q $ and $ T(q_1,q_2) = T_{q_1q_2} $ for $ (q_1,q_2) \in  \{(1,2),(2,1),(2,3),(3,2)\} $. \qed
\end{example}

\begin{figure}
	\centering
	\includegraphics[width=0.8\linewidth]{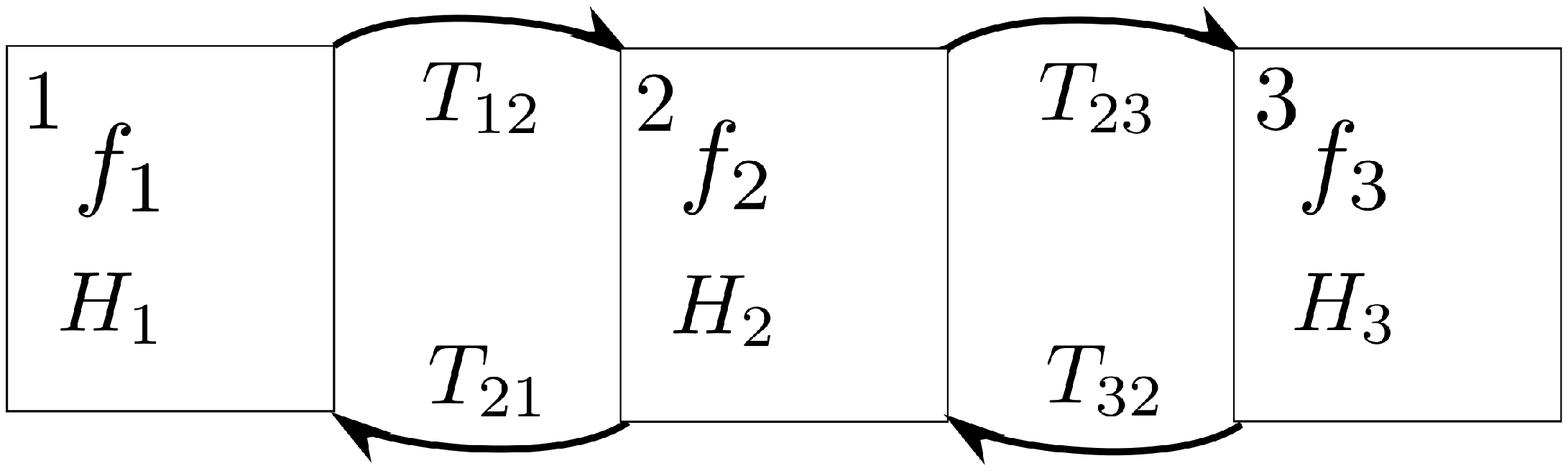}
	\caption{This preview automaton corresponds to the switched system in Example \ref{eg:1}.}
	\label{fig:previewautomata}
\end{figure}

Any transition in the preview automaton is associated with an input in the form of the preview of the switching mode. We assume that there is at most one preview between any two consecutive switches. During the execution of the preview automaton, if a preview takes place at time $ t $, there is a corresponding \emph{preview input} of the preview automaton, including the timestamp $ t $ of the occurrence of the preview, the destination state $ d\in Q $ of the next transition and the remaining time steps (the preview time) $ \tau $ from the current time $ t $ up to the next transition.  If no preview takes place before the next switching time\footnote{This is possible when the lower bound of the time interval of possible preview times is $ 0 $.}, the preview input corresponding to that switch is trivially $ (t, 0, q) $, where $ t $ and $ q $ are the time instant and destination of the next transition. Note that $ t+\tau $ is the time that the system transits from the last mode to the mode $ d $.

\begin{definition}Given preview automaton $ G=\{ Q,E,T,H\}  $ and initial state $ q_0\in Q $, a sequence of tuples $ \{(t_{k},\tau_{k}, d_{k})\}_{k=1}^{N}$ ($N < \infty$ when the system remains in $ d_{N} $ after $ t\geq t_{N}+\tau_{N} $) is a valid \emph{preview input sequence} of $G$ if for all {\color{black} $1\leq k\leq N $ }, the sequence satisfies (with $ t_0=0 $, $ \tau_0=0 $, $ d_0=q_0 $) that \\
(1) $ \tau_{k-1}\geq 0 $ and  $ t_{k-1}+\tau_{k-1}\leq t_{k} $ and\\
(2) $ (d_{k-1},d_k)\in E $ and $ \tau_{k}\in T(d_{k-1},d_{k}) $ and\\
(3) $ (t_{k}+\tau_k-1) -( t_{k-1}+\tau_{k-1})\geq H(d_{k-1})$.\label{def:input_pa}
\end{definition}
In above definition, conditions (1), (2) and (3) guarantee that only one preview input is received between two consecutive switches, the mode switch constraints and preview time constraints are met, and the holding time constraint is met, respectively. Once a valid input sequence is given, we can uniquely identify the transitions of the preview automaton over time, that is, the execution of the preview automaton with respect to that input sequence. In the rest of the paper, we only consider valid preview input sequences and drop the word valid when it is clear from the context.

\begin{definition}
Given preview automaton $G = \{Q,E,T,H\}$ and a preview input sequence $ \{t_{k},\tau_{k}, d_{k}\}_{k=1}^{N}$, the \emph{execution} of $G$ with respect to the preview input sequence is a sequence of tuples $\{( I_{k},q_{k})\}_{k=0}^{N}$, where \\
(1) $ I_0=[0,t_1+\tau_1 -1] $ and $ I_k=[t_k+\tau_k,t_{k+1}+\tau_{k+1}-1]$ for all $ k\geq 0 $, \\
(2) $q_k=d_k $ for all $ k\geq 1 $.  
\label{def:run_pa} 
\end{definition}
Note that two different valid preview input sequence may have the same execution. According to Definition \ref{def:input_pa} and \ref{def:run_pa}, the set of possible executions of one preview automaton is determined by the set of valid preview input sequences of the preview automaton. 

Once we have the preview automaton $ G $ corresponding to a switched system $ \Sigma $, we have a model of the allowable switching sequences for $ \Sigma $, given by the executions of $ G $. Therefore we can define the runs of a switched system with respect to a preview automaton.

\begin{definition}	A sequence $ \{(q(t),x(t))\}_{t=0}^{\infty} $ is a \emph{run of the switched system} $ \Sigma$ of $ s $ modes with preview automaton $ G $ under the control inputs $ \{u(t)\}_{t=0}^{\infty} $ if (1) $ \{q(t)\}_{t=0}^{\infty} $ is an execution of $ G $ for some preview input sequence and (2) $ x(t+1)\in f_{q(t)}(x(t),u(t)) $ for $ t\geq0 $.
\end{definition}

\subsection{Problem Statement}\label{sec:prob}
Though the preview automaton can be useful in the existence of more general specification, we focus only on safety specifications in this work. Suppose that each mode $ k $ is associated with a \emph{safe set} $ S_k\subseteq X $, that is the set of states where we require the switched system to stay within when the system's active mode is $ k $. Denote the collection of safe sets $ \{S_i\}_{i\in Q} $ for each mode as a \emph{safety specification} for the switched system $ \Sigma $. Then, given a safety specification $ \{S_i\}_{i\in Q} $, a run $ \{(q(t),x(t))\}_{t=0}^{\infty} $ is \emph{safe} if $ (q(t),x(t))\in \cup_{i\in Q} (i,S_i) $ for all $ t\geq 0 $. Otherwise, this run is \emph{unsafe}.

A controller is usually assumed to know the partial run of the system up to the current time before making a control decision at each time instant. In our case, since the system can look ahead and see the next transition, reflected by the preview input signal, the controller for a switched system equipped with a preview automaton is assumed to have access to the preview inputs of the preview automaton up to the current time.

\begin{definition} Denote  $ \{(p(t),x(t))\}_{t=0}^{t^*} $ and $ \{(t_{k},\tau_{k},d_{k})\}_{k=0}^{k^*} $ as the partial run and preview inputs of the switched system up to time $ t^* $ ($ k^* $ refers to the latest preview up to $ t^* $, {\color{black}i.e., $k^* = \max{k} \text { s.t. } t_k\leq t^*$}). A \emph{controller} $ \mathcal{U} $ of the switched system $ \Sigma $ with preview automaton $ G $ is a function that maps the partial run  $ \{(p(t),x(t))\}_{t=0}^{t^*} $ and the preview inputs $ \{(t_{k},\tau_{k},d_{k})\}_{k=0}^{k^*} $  to a control input $ u(t^*) $ of the switched system for any $ t^*\geq 0 $.\label{def:controller}
\end{definition}

\begin{definition}	Given a switched system $ \Sigma $ and a safety specification $ \{S_i\}_{i\in Q} $, 
	a subset $ W_i $ of the state space of $\Sigma$  is a single \emph{winning set with respect to mode $ i $} if there exists a controller $ \mathcal{U} $ such that any run of the closed-loop switched system with initial state in $ \{i\}\times W_i $ is safe. A winning set $ W_i $ is the \emph{maximal winning set with respect to the mode $ i $} if for any $ x\not\in W_i $, for any controller $ \mathcal{U} $, there exists an unsafe run with initial state $ (i,x) $. A \emph{winning set with respect to the switched system $ \Sigma $} is the collection $ \{W_i\}_{i=1}^{n} $ of single winning sets for all modes, which is called \emph{winning set} for short. \label{def:win_set}
\end{definition}

We note that, by definition, arbitrary unions of winning sets with respect to one mode is still a winning set, and therefore the maximal winning set is unique under mild conditions \cite{bertsekas1972infinite}
 and contains all the winning sets with respect to that mode. Now, we are ready to state the problem of interest.

\begin{problem}
	Given a switched system $ \Sigma $ with corresponding preview automaton and safety specification $ \{S_i\}_{i\in Q} $, find the maximal winning set $ \{W_i\}_{i\in Q} $. \label{prob:1}
\end{problem}

\begin{figure}
	\centering
	\includegraphics[width=0.8\linewidth]{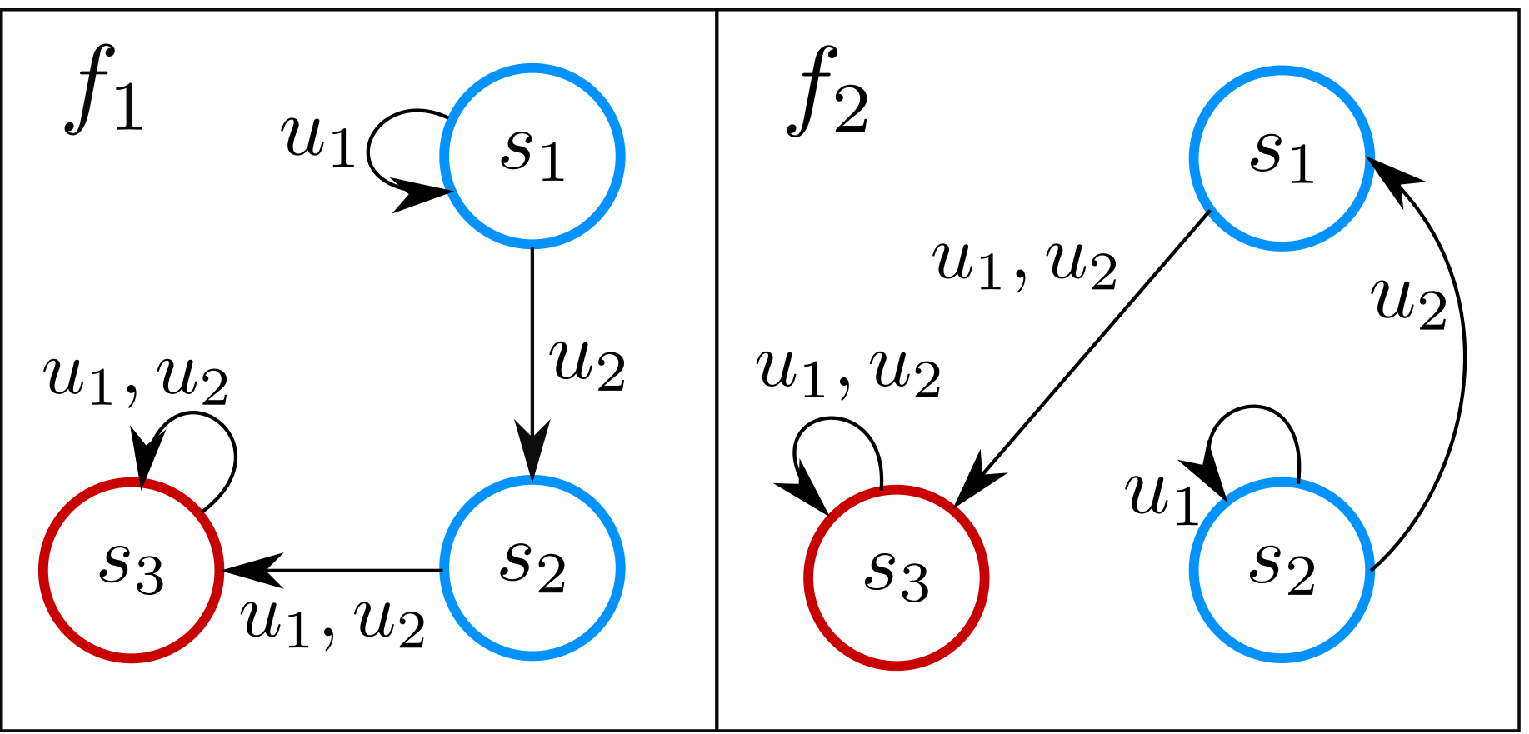}
	\caption{A switched finite transition system with $ 2 $ modes $ f_1 $ and $ f_2 $. The safe set (blue) is $ \{s_1,s_2\} $ for each mode.}
	\label{fig:toydyn}
\end{figure}

Before an algorithm that computes the maximal winning set is introduced, we first study the following toy example to demonstrate the usefulness of preview information. 

\begin{example}\label{eg:fts}
	A switched transition system with two modes is shown in Figure \ref{fig:toydyn}. The state space and input space of the switched system are $ \{s_1,s_2,s_3\} $ and $ \{u_1,u_2\} $ respectively. The safety specification is $ S_1 =S_2=\{s_1,s_2\}$. To satisfy this safety specification, the system state has to be $ s_1 $ when $ f_1 $ is active and be $ s_2 $ when $ f_2 $ is active. Thus by inspection, when there is no preview, the winning sets are empty and when there is a preview for at least one-step ahead before each transition, there is a non-empty winning set $ W_1= \{s_1\} $ and $ W_2=\{s_2\} $. \qed
\end{example}

Example \ref{eg:fts} suggests that a winning set is not the same as a controlled invariant set of the switched system. When the preview information is ignored or unavailable, they are the same and therefore Problem \ref{prob:1} can be solved by computing the controlled invariant sets within the safe sets. However, if the preview is available, the controlled invariant sets can be conservative since their computation does not take advantage of the online preview information. Therefore, in Example \ref{eg:fts}, the maximal controlled invariant set is empty, but the winning set is non-empty.

%% file: sec_example.tex
In the following case studies, we apply the proposed algorithms to switched affine systems, where the state space and safe set are polytopes. In this case $Pre$ and $PreInv$ operators reduces to polytopic operations, which we implement using the MPT3 toolbox \cite{MPT3}.

\subsection{Vehicle Cruise Control}
Our first example is a cruise control problem for the scenario shown in Example \ref{eg:1}. The longitudinal dynamics of a vehicle with road grade is given by
\begin{align}
\dot{v} = -\frac{f_0}{m}  -\frac{f_1}{m}v + \frac{F_w}{m} - g\sin\theta\label{eqn:dyn_cc}
\end{align}
where $ v $ is the longitudinal speed, $ m $ is the vehicle mass, $ f_0 $ and $ f_1 $ are the coefficients related to frictions, $ F_w $ is the wheel force, $ g $ is the gravitational acceleration and $ \theta $ is the road grade. We choose $ F_w $ as the control input and $ \theta $ as a disturbance.  We discretize \eqref{eqn:dyn_cc} with time step $ \Delta t= 0.1s $. The discrete-time dynamics with disturbance ranges $ r_1$, $ r_2$ and $ r_3 $ consist of the modes $ 1 $, $ 2 $, $3 $ in the switched system defined in Example \ref{eg:1}.

The safety specification is to keep the longitudinal speed within $ X=[31.95 ,32 ]m/s $. The speed range is intentionally picked small enough so that the change the road grade induces on the dynamics makes the specification hard to be satisfied. The parameters are chosen as $ m=1650 kg$, $ f_0 = 0.1 N $, $ f_1 = 5 N\cdot s/m $, $ g = 10m/s^2 $. The control input range is $ F_w\in [-0.65 mg,0.66mg] $. For the preview automaton shown in Fig. \ref{fig:previewautomata}, the holding time for each mode is $ 2 $ and the preview time for each transition is $ 1 $.

To make a comparison, we compute the maximal controlled invariant set for the dynamics discretized from \eqref{eqn:dyn_cc} with disturbance in $ [-30.5^\circ,30.5^\circ] $ (convex hull of $ r_1 $, $r_2 $, $r_3 $). If such an invariant set exists, it is a feasible winning set for our problem. However, the resulting controlled invariant set is empty, which suggests that the problem is infeasible if disturbance can vary arbitrarily in $ [-30.5^\circ, 30.5^\circ] $. In contrast, the winning set obtained from Algorithm \ref{alg:md2_nd} is $\{ W_i\}_{i=1}^3 $ with $ W_1 = W_2 = W_3 = X $. Therefore the preview automaton is crucial in this case study for the existence of a safety controller.
\subsection{Vehicle Lane Keeping Control}

In the second example, we apply the proposed method to synthesize a lane-keeping controller, which controls the steering to limit the lateral displacement of vehicle within the lane boundaries. 

The lateral dynamics we use are from a linearized bicycle model\cite{smith}. The four states of the model consist of the lateral displacement $ y $, lateral velocity $ v $, yaw angle $ \Delta\Psi $ and yaw rate $ r $.  The vehicle is controlled by the steering input $ \delta_f $ in range $ [-\pi/2,\pi/2] $. We assume that the longitudinal velocity $u$ of the vehicle is constant and equal to $30 m / s$. The disturbance $r_d $ is a function of the road curvature, which is what we assume to have preview information on at run-time.

The maximal recommended range of $ r_d $ on Michigan highways\cite{michiganroad} with respect to $u=30m / s$ is about $ [-0.06,0.06] $. We divide $ [-0.06,0.06] $ evenly into $ 5 $ intervals $ d_1 = [-0.06,-0.036]$, $d_2=[-0.036,-0.012]$, ..., $d_5=[0.036,0.06] $ and construct a switched system with $ 5 $ modes, where each mode $i\in Q =\{1,2,3,4,5\} $ corresponds to a lateral dynamics with $ r_d $ bounded in $ d_i $, denoted by $f_i$. The corresponding preview automaton is shown in Fig. \ref{fig:pacasestudy}, where transitions are only between any two modes with  adjacent $ r_d  $ intervals. For simplicity, the preview time interval $T\left( i,j \right) = \tau_c $ for all $\left( i,j \right) \in E$, and the least holding time $H\left( i \right) = \tau_d$ for all $i\in Q$ for some constants $\tau_c$ and $\tau_d$.

The safe set is given by the constraints  $\vert y \vert \leq 0.9$, $\vert v\vert \leq 1.2$, $\vert \Delta \Psi \vert \leq 0.05$, $\vert r\vert \leq 0.3 $ for all modes. 

\begin{table}[h]
	\centering
	\caption{Computation costs for different $(\tau_c,\tau_d)$}
	\label{tab:comp_cost}
	\begin{tabular}{ccc}
		\hline
		$ (\tau_c,\tau_d) $ & \#iterations & time (min)\\ 
		\hline
		$(1,2)$ & $4$ & $18.9$ \\
		$ (2,2)$  & $4$ & $18.0$ \\ 
		$(1,1)$  & $5$ & $20.3$ \\ 
		$(5,5)$  & $3$ & $16.8$ \\  
		\hline
	\end{tabular}
\end{table}

We apply Algorithm \ref{alg:md2_nd} to compute the maximal winning sets for various $\tau_{c}$ and $\tau_d$. The values of $(\tau_c,\tau_d)$ with corresponding numbers of iterations at termination and running time are listed in Table \ref{tab:comp_cost}.  Denote the maximal winning set with respect to mode $2$ for each pair $(\tau_c,\tau_d)$ in Table \ref{tab:comp_cost} as $ W_{2,(\tau_c,\tau_d)} $. 

\begin{figure}
	\centering
	\includegraphics[width=0.7\linewidth]{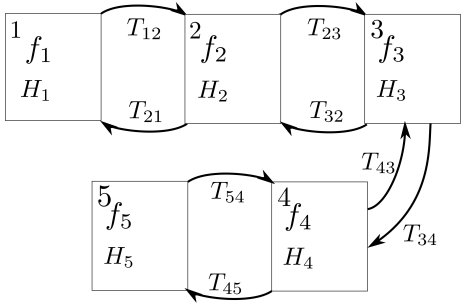}
	\caption{Preview automaton for the lane-keeping case study}
	\label{fig:pacasestudy}
\end{figure}

As a comparison, we compute the maximal controlled invariant set for the lateral dynamics with $ r_d $ in $ [-0.06,0.06] $, denoted by $ W_{inv} $. The projections of $ W_{2,(\tau_c,\tau_d)} $ and $ W_{inv} $ onto $3$-dimensional subspaces are shown in Figs. \ref{fig:proj_H2T12} and \ref{fig:proj_T1H15}.

Fig. \ref{fig:proj_H2T12} compares $W_{inv}$, $ W_{2,(1,2)}$ and $W_{2,(2,2)}$, where the holding time $\tau_d$ is fixed and the preview time $\tau_c$ are tuned to show the effect of preview time on winning set. In theory, $ W_{inv}\subseteq W_{(\tau_c,\tau_d)}\subseteq W_{(\tau_c',\tau_d')} $ for any $\tau_c \leq  \tau_c'$ and $\tau_d \leq  \tau_d'$, which is verified by the numerical result where $W_{inv}\subseteq W_{2,(1,2)} \subseteq W_{2,(2,2)}$. The blue region in Fig. \ref{fig:proj_H2T12} shows the difference of $ W_{2,(1,2)} $ and $ W_{inv} $, indicating how much we gain from the preview information with $(\tau_c,\tau_d) = (1,2)$ compared to no preview. The green region in Figure \ref{fig:proj_H2T12} shows the difference of $W_{2,(2,2)}$ and $W_{2,(1,2)}$, which indicates how much the maximal winning set grows as the preview time $ \tau_c $ increases from $1$ to $2$ while the least holding time $\tau_d = 2$ is fixed. As revealed by the size of the green region in Fig. \ref{fig:proj_H2T12}, the growth of the maximal winning set decreases as the preview time becomes one step longer. Understanding the conditions under which a longer preview does or does not help the growth of the maximal winning set is subject of our future work.  

\begin{figure}
	\centering
	\begin{subfigure}{0.49\linewidth}
		\centering
		\includegraphics[width=1\linewidth]{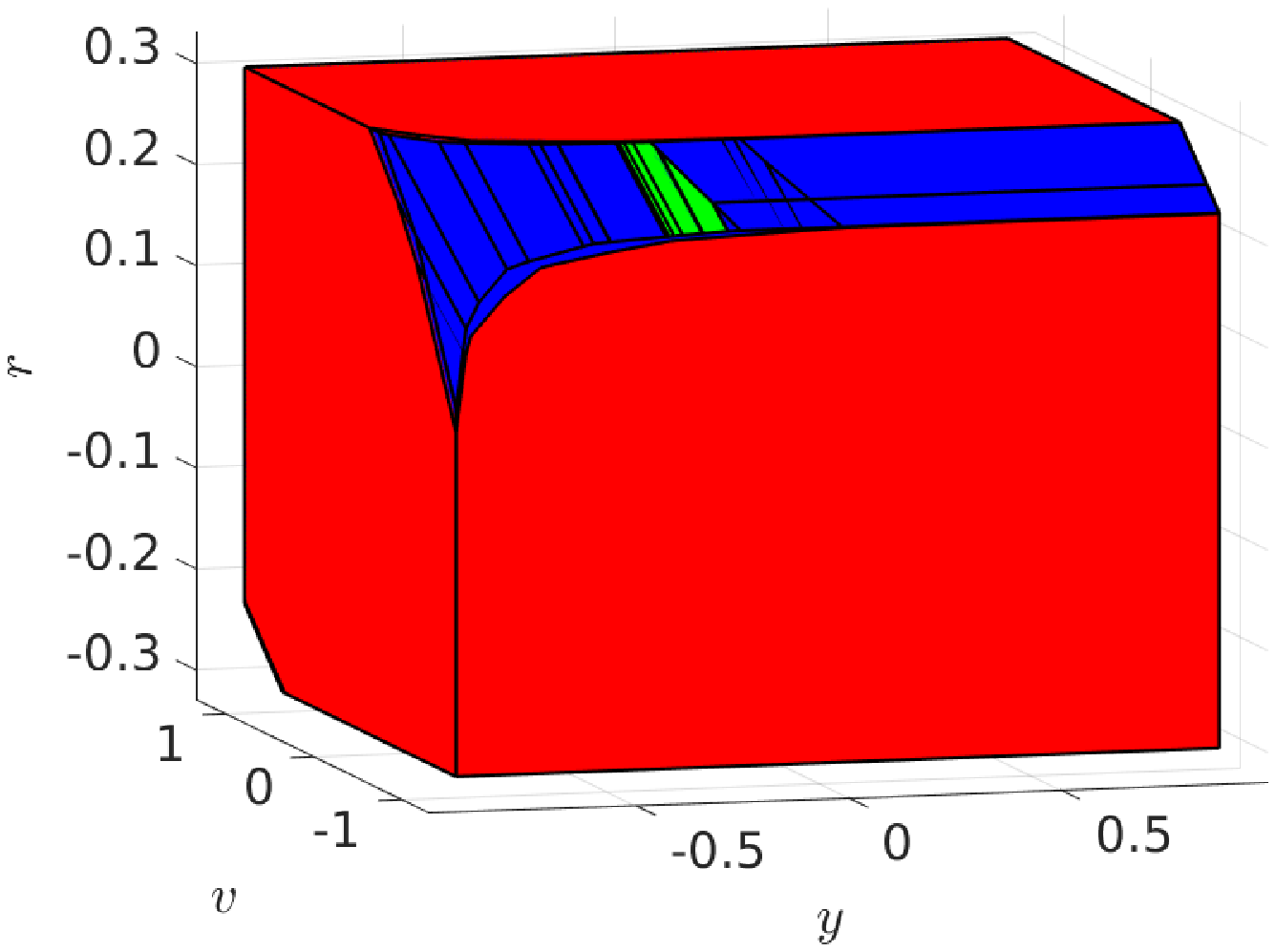}
		\caption{project on $ (y,v,\Delta \Psi) $}
		\label{fig:proj1_H2T12}
	\end{subfigure}
	\begin{subfigure}{0.49\linewidth}
		\centering
		\includegraphics[width=1\linewidth]{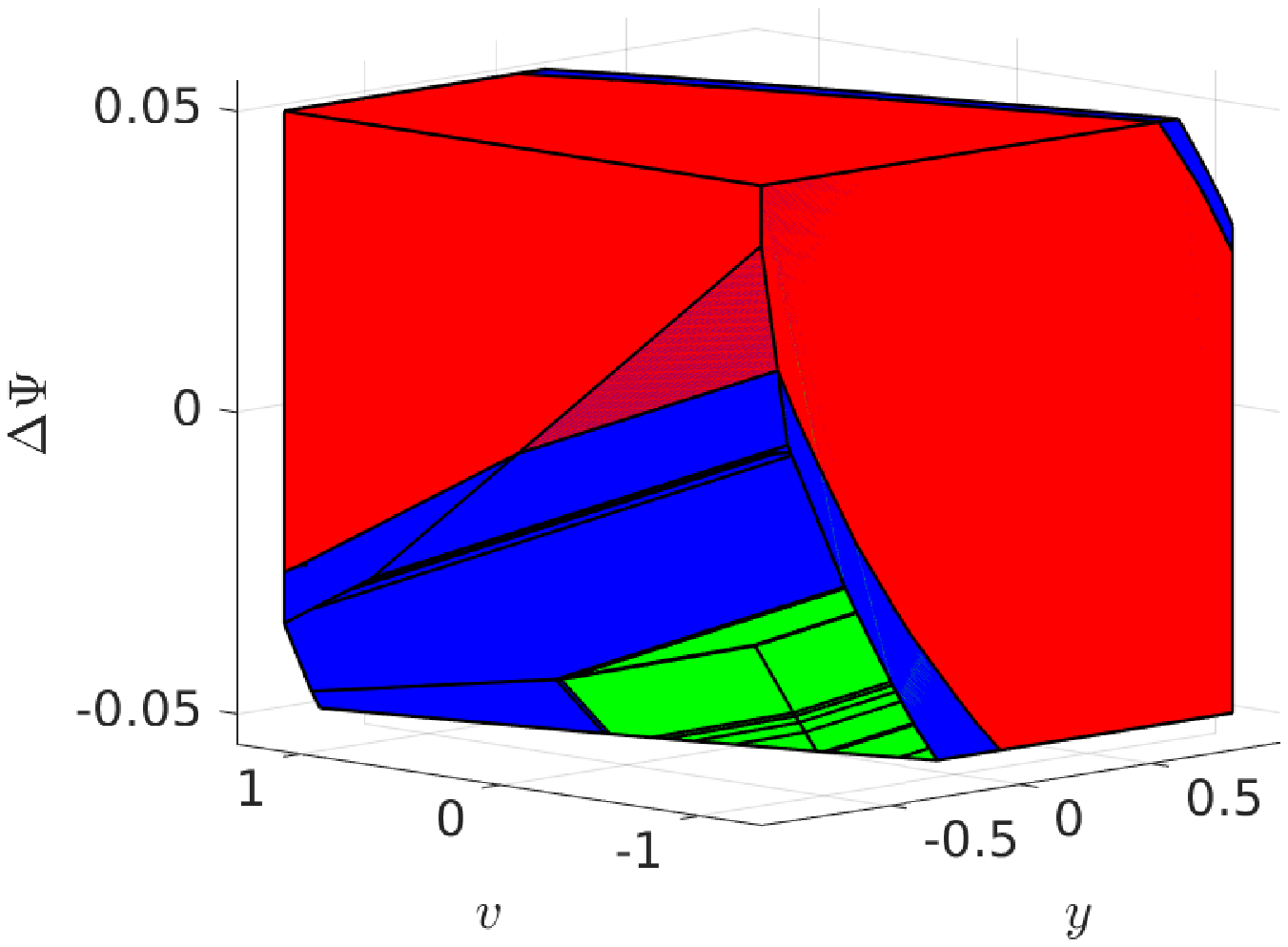}
		\caption{project on $ (y,v,r) $ }
		\label{fig:proj2_H2T12}		
	\end{subfigure}
	\caption{Projections of $W_{inv}, W_{2,(1,2)}, W_{2,(2,2)}$ onto two subspaces. The red, blue and green regions are the projection of  $W_{inv}$, the difference of projections of $W_{2,(1,2)}$ and $W_{inv}$ and the difference of projections of $W_{2,(2,2)}$ and $W_{2,(1,2)}$. }
	\label{fig:proj_H2T12}
\end{figure}

Fig. \ref{fig:proj_T1H15} compares $W_{inv}$, $W_{2,(1,1)}$ and $W_{2,(1,5)}$, where we fix the preview time $\tau_c$ and change the least holding time $\tau_d$. The blue and green regions show the difference of $W_{2,(1,1)}$ and $W_{inv}$ and the difference of $W_{2,(1,5)}$ and $W_{2,(1,1)}$. Therefore, the size of the green region indicates how much the winning set grows as we increase the least holding time $\tau_d$ from $1$ to $5$. Compared to Fig. \ref{fig:proj_H2T12}, the winning set is more sensitive to the change of the least holding time  $\tau_d$ than the change of the preview time $\tau_c$. 

Finally, $W_{inv} $, $W_{2,(1,1)}$, $W_{2,(5,5)}$ are compared in Fig. \ref{fig:proj_T1H15}, where we increase $\tau_c$ and $\tau_d$ simutaneously. $W_{2,(5,5)}$ is numerically equal to  $W_{2,(1,5)}$, and thus its projections are the same as the projections of $W_{2,(1,5)}$ shown in Fig. \ref{fig:proj_T1H15}.  In fact, the winning sets with respect to modes $2$, $3$, $4$ for $\tau_c=1,\tau_d=5$ and $\tau_c=5,\tau_d=5$ are numerically equal; the winning set with respect to mode  $1$ and $5$ slightly grows when  $(\tau_c,\tau_d)$ changes from $(1,5)$ to $(5,5)$, but the growth is too small to be visualized. The observation in Fig. \ref{fig:proj_H2T12} and \ref{fig:proj_T1H15} reveals one theoretical conjucture: If the preview time and the least holding time are large enough, a longer preview time and/or a longer holding time will not increase the size of the maximal winning set. That is, the size of the maximal winning set converges as the preview time and the least holding time increase. To verify this conjecture is part of our future work.

\begin{figure}
    \begin{subfigure}{0.49\linewidth}
		\centering
		\includegraphics[width=1\linewidth]{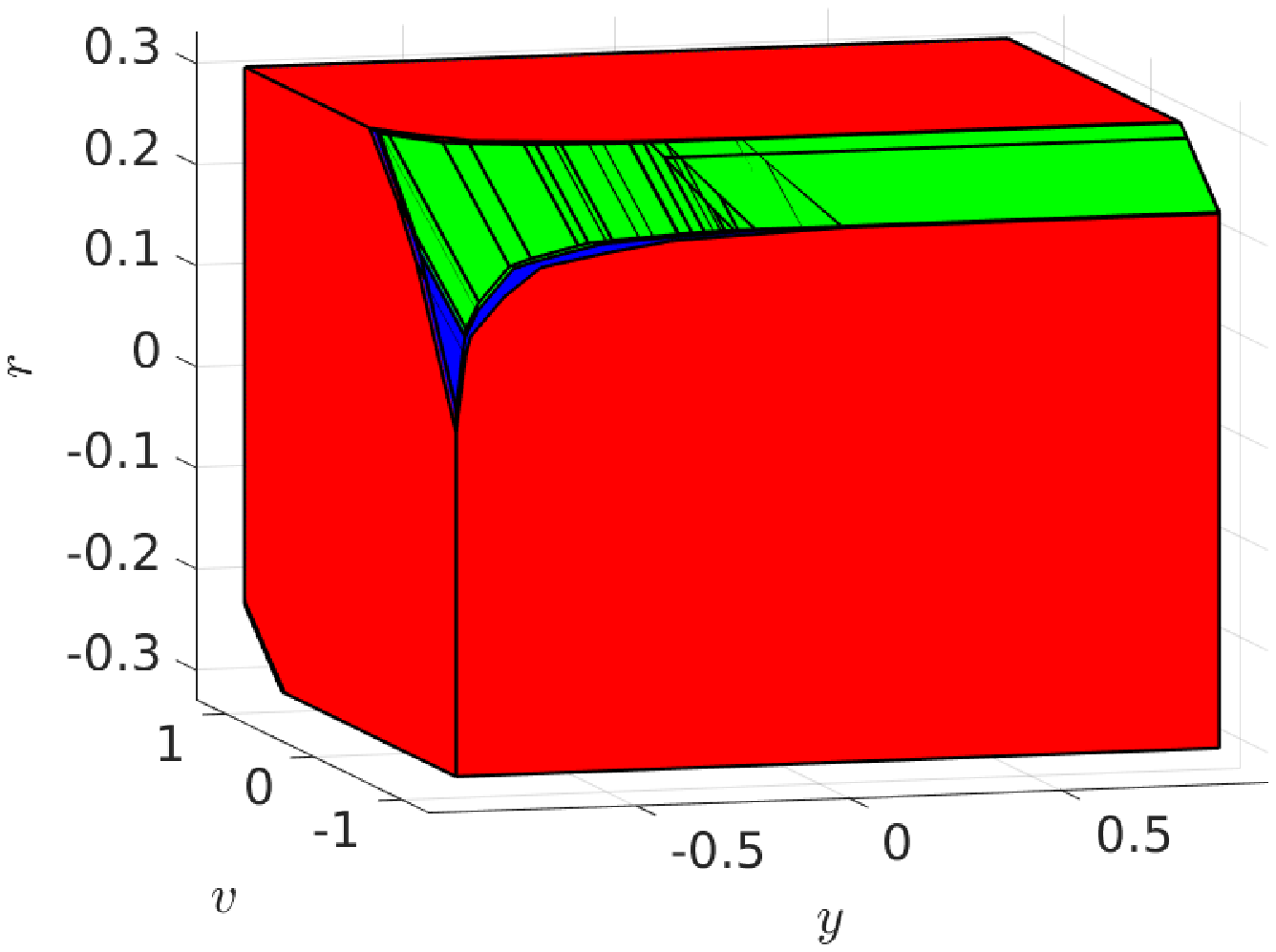}
		\caption{project on $ (y,v,\Delta \Psi) $ }
		\label{fig:proj1_T1H15}
	\end{subfigure}
	\begin{subfigure}{0.49\linewidth}
		\centering
		\includegraphics[width=1\linewidth]{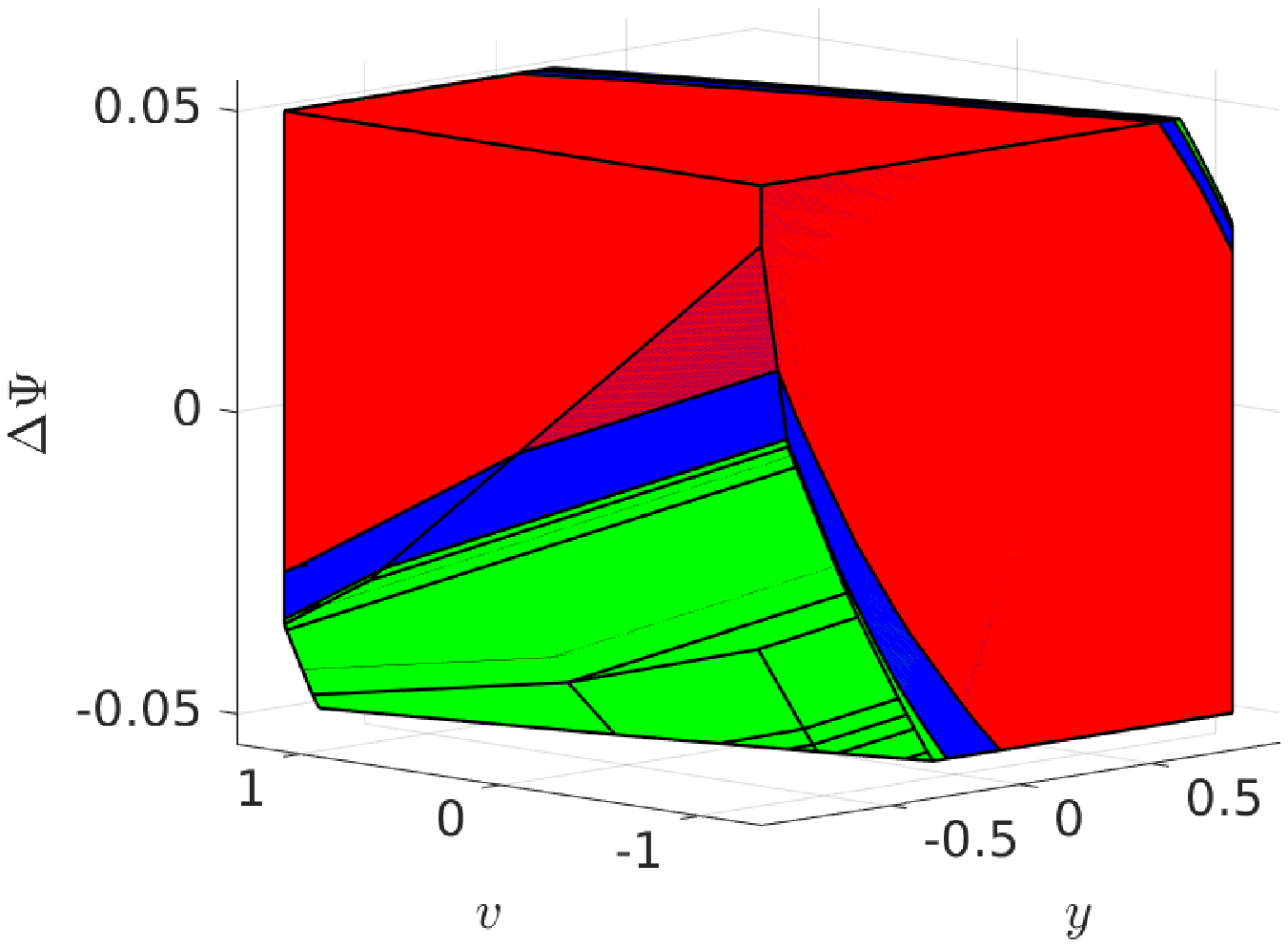}
		\caption{project on $ (y,v,r) $ }
		\label{fig:proj2_T1H15}		
	\end{subfigure}
	\caption{Projections of $W_{inv}, W_{2,(1,1)}, W_{2,(1,5)} $ onto two subspaces. The red, blue and green regions are the projection of  $W_{inv}$, the difference of projections of $W_{2,(1,1)}$ and $W_{inv}$ and the difference of projections of $W_{2,(1,5)}$ and $W_{2,(1,1)}$. }
	\label{fig:proj_T1H15}
\end{figure}

%% file: sec_app.tex
\begin{proof}[Proof of Lemma \ref{lemma:1}]
	By definition of $ Pre $, for any sets $ W_1\subseteq W_2 \subseteq S$, $ Pre^{f_1}(W_1)\subseteq Pre^{f_1}(W_2)$ for any $ i\in Q $. Therefore $ PreInt^{f_i} (W_1,S_i) \subseteq PreInt^{f_i} (W_2,S_i)\subseteq S_i$ and furthermore $ Inv^{f_i}(W_1)\subseteq Inv^{f_i}(W_2)\subseteq S_i $ for any $ i\in Q $. According to line \ref{line:init} of Algorithm \ref{alg:md2_nd}, the input $ W $ of $ InvPre $ is always a subset of $ S $. Note that all the intermediate variables $ C_{l,j}$, $C_{T_{min}}$, $C_k $ are recursively computed by $ PreInt $ and $ Inv $. Based on the monotonicity of $ PreInt $ and $ Inv $, we can check step by step that the values of the intermediate variables $ C_{l,j}$, $C_{T_{min}}$, $C_k $ of $ InvPre $ with respect to input $ W_1 $ are contained by the values of those variables with respect to input $ W_2 $. Therefore $ InvPre^{f_i}(G,W_1,S) \subseteq InvPre^{f_i}(G,W_2,S)\subseteq S_i$.
\end{proof}

\begin{proof}[Proof of Lemma \ref{lemma:2}]
	First note that $ InvPre^{f_i} \{ G, W, S\} $ only depends on $ G $, $ \{W_j\}_{j\in Post^G(i)} $ and $ S $, though we use the whole $ W $ instead of $ \{W_j\}_{j\in Post^G(i)} $ as input of Algorithm \ref{alg:invprend} for short. In this proof, we change the notation to $ InvPre^{f_i}(G,\{W_j\}_{j\in Post^G(i)}, S) $ to make this point clear.

The key observation in this proof is: if the system switches from node $ i $ to node $ j $ at some time $ t $ with state $ x(t) $, for the purpose of synthesizing future control strategies, it is equivalent to the case that the system initially starts from the state $ (j,x(t)) $, and therefore there exists a controller to guarantee safety for the rest of the run if and only if $ x(t)\in W^*_j $.
 
Denote $ \{W^*_i\}_{i\in Q} $ as the maximal winning set. By Proposition \ref{prop:max}, $ W^*_j=Inv^{f_j}(S_j) $ for all sink states $ j\in Q_s $. Let $ k\in  Q_{ns} $. Now suppose that we know the maximal winning set $ W^*_k $ with respect to all $ k\in  Q_{ns} $ except $ i $. We want to show that $ W^*_i $ can actually be computed by $ InvPre^{f_i}(G,\{W^*_j\}_{j\in Post(i)},S) $. Note that $ i\not\in Post(i) $ since we do not allow self-loops in the preview automaton.

Denote the minimum preview time among all feasible transitions as $ T_{min} = \min_{j\in Post^G(i)}T(i,j) $. Let us first consider the case where no preview happens during the time interval $ [0,t-1] $ with $ t\geq H(i)-T_{min} $. In this case, the least holding time constraint will be satisfied whatever the next transition is. Let us consider the maximal set of states $ C_{T_{min}} $ such that if $ x(t) $ is within $ C_{T_{min}} $, there exists a controller that makes the closed-loop system satisfy the safety spec.

Suppose that one preview happens at $ t $ with destination state $ j\in Post^G(i) $ and remaining time $ \tau_{ij} =T(i,j) $. To guarantee the safety spec being satisfied in the future, for any $ t\leq t'< t+\tau_{ij} $, $ x(t') $ has to be within $ S_i $, and $ x(t+\tau_{ij}) $ has to be within $ W^*_j$ at time $ t+\tau_{ij} $. The maximal subset of $ S_i$ that is able to reach $ W^*_j $ in one step is given by $ PreInv^{f_i}(W^*,S) $. By applying $ PreInv $ $ \tau_{ij} $ many times, we obtain the maximal subset of states in $ S_i $ that can stay within $ S_i $ until reaching $ W^*_j $ at $ t+\tau_{ij} $, which is $ C_{T_{ij},j} $ in line \ref{line:for_C_ij}-\ref{line:end_for_C_ij} of Algorithm \ref{alg:invprend}. The intersection $ \bigcap_{j\in Post^G(i)} C_{T_{ij}} $ is the maximal set of values of $ x(t) $ that guarantees safety spec under the condition that a preview takes place at $ t $.  If there is no preview at $ t $, $ t+1 $, $ ... $, the system needs to stay within $ \bigcap_{j\in Post^G(i)} C_{T_{ij}} $ for safety at time $ t+1 $, $t+2 $, ... Therefore $ C_{T_{min}}  = Inv^{f_i}(\bigcap_{j\in Post^G(i)} C_{T_{ij}})$ (line 9).

Now consider the case where no preview takes place between $ [0,t-1] $ where $ t = H(i)-T_{min}-1 $. Denote $ C_{T_{min}+1} $ as the maximal subset of $ S_i $ such that if $ x(t) \in C_{T_{min}+1}$ there exists a controller to guarantee safety spec; otherwise there is no such a controller. If there is no preview at time $ t $, all we want for safety is $ x(t+1)\in C_{T_{min}} $. Then $ PreInv^{f_i}(C_{T_{min}}) $ is the maximal set of values of $ x(t) $ such that there exists a control input to make sure $ x(t+1)\in C_{T_{min}} $. Otherwise, if there is a preview at $ t $ with destination state $ j $, all we want is $ x(t)\in C_{T_{ij}} $ by the previous discussion. 
The set of all feasible destination states that can have preview at $ t $ is $ J_{T_{min}+1} =\{j\in Post^G(i):T_{ij}\geq T_{min}+1\} $. Therefore, depending on whether $ J_{T_{min}+1} $ is empty or not, $ C_{T_{min}+1} $ is equal to $ PreInv^{f_i}(C_{T_{min}})  $ or $ PreInv^{f_i}(C_{T_{min}}) \cap (\bigcap_{j\in J_{T_{min}+1}} C_{T_{ij}}) $ (line \ref{line:C_k_1} and \ref{line:C_k_2}), which is guaranteed to have a safety controller for all the situations. The same discussion can be repeated for $ t =  H(i)-T_{min}- k $ up to $ k = H(i)-T_{min} $, resulting in $ C_{T_{min}+k} = PreInv^{f_i}(C_{T_{min}+k-1}) \cap \bigcap_{j\in J_{T_{min}+k}} C_{T_{ij}}  $ for $ J_{T_{min}+k} = \{j\in Post^G(i):T_{ij}\geq T_{min}+k\} $. For the case $ k = H(i) - T_{min} $, $ t=0 $ and $ C_{H(i)} $ is the maximal set of values of $ x(0) $ such that if $ x(0)\in C_{H(i)} $, there exists a controller that makes sure the closed-loop system satisfy the safety spec, which is the maximal winning set with respect to $ i $. Therefore $ W^*_i = C_{H(i)} = InvPre^{f_i}(G,\{W^*_j\}_{j\in Post^G(i)},S) $ (line \ref{line:C_k_2}), and $ \{W^*_j\}_{j\in  Q_{ns}} $ is a solution of equations in \eqref{eqn:invpre}.

We have proved that the maximal winning set must satisfy the equations in \eqref{eqn:invpre}. On the other hand, by construction of $ InvPre $, it can be verified that any solution $ \{W_i\}_{i\in  Q_{ns}} $ with $ W_j\subseteq S_j $ for $ j\in  Q_{ns} $ of the equations in \eqref{eqn:invpre} with respect to $ \{W^*_k\}_{k\in Q_s} $ is a winning set (not necessarily maximal). Also, since by definition, the union of winning sets is still a winning set, the maximal winning set must be unique and contain all the winning sets. Therefore, $ \{W^*_i\}_{i\in Q_{ns}} $ must be the maximal solution of equations in \eqref{eqn:invpre} with respect to $ \{W^*_k\}_{k\in Q_s} $. 
\end{proof}

%% file: ms.bbl
\begin{thebibliography}{10}
\providecommand{\url}[1]{#1}
\csname url@samestyle\endcsname
\providecommand{\newblock}{\relax}
\providecommand{\bibinfo}[2]{#2}
\providecommand{\BIBentrySTDinterwordspacing}{\spaceskip=0pt\relax}
\providecommand{\BIBentryALTinterwordstretchfactor}{4}
\providecommand{\BIBentryALTinterwordspacing}{\spaceskip=\fontdimen2\font plus
\BIBentryALTinterwordstretchfactor\fontdimen3\font minus
  \fontdimen4\font\relax}
\providecommand{\BIBforeignlanguage}[2]{{%
\expandafter\ifx\csname l@#1\endcsname\relax
\typeout{** WARNING: IEEEtran.bst: No hyphenation pattern has been}%
\typeout{** loaded for the language `#1'. Using the pattern for}%
\typeout{** the default language instead.}%
\else
\language=\csname l@#1\endcsname
\fi
#2}}
\providecommand{\BIBdecl}{\relax}
\BIBdecl

\bibitem{zexiang2019poster}
Z.~Liu and N.~Ozay, ``Safety control with preview automaton,'' in
  \emph{Proceedings of the 22nd ACM International Conference on Hybrid Systems:
  Computation and Control}.\hskip 1em plus 0.5em minus 0.4em\relax ACM, 2019,
  pp. 280--281.

\bibitem{sheridan1966three}
T.~B. Sheridan, ``Three models of preview control,'' \emph{IEEE Transactions on
  Human Factors in Electronics}, no.~2, pp. 91--102, 1966.

\bibitem{tomizuka1975optimal}
M.~Tomizuka and D.~Whitney, ``Optimal discrete finite preview problems (why and
  how is future information important?),'' \emph{Journal of Dynamic Systems,
  Measurement, and Control}, vol.~97, no.~4, pp. 319--325, 1975.

\bibitem{katayama1985design}
T.~Katayama, T.~Ohki, T.~Inoue, and T.~Kato, ``Design of an optimal controller
  for a discrete-time system subject to previewable demand,''
  \emph{International Journal of Control}, vol.~41, no.~3, pp. 677--699, 1985.

\bibitem{hazell2008discrete}
A.~Hazell, ``Discrete-time optimal preview control,'' Ph.D. dissertation,
  Imperial College London, 2008.

\bibitem{peng1993preview}
H.~Peng and M.~Tomizuka, ``Preview control for vehicle lateral guidance in
  highway automation,'' \emph{Journal of Dynamic Systems, Measurement, and
  Control}, vol. 115, no.~4, pp. 679--686, 1993.

\bibitem{xu2019design}
S.~Xu and H.~Peng, ``Design, analysis, and experiments of preview path tracking
  control for autonomous vehicles,'' \emph{IEEE Transactions on Intelligent
  Transportation Systems}, 2019.

\bibitem{garcia1989model}
C.~E. Garcia, D.~M. Prett, and M.~Morari, ``Model predictive control: theory
  and practice -- a survey,'' \emph{Automatica}, vol.~25, no.~3, pp. 335--348,
  1989.

\bibitem{laks2011model}
J.~Laks, L.~Pao, E.~Simley, A.~Wright, N.~Kelley, and B.~Jonkman, ``Model
  predictive control using preview measurements from lidar,'' in \emph{49th
  AIAA Aerospace Sciences Meeting including the New Horizons Forum and
  Aerospace Exposition}, 2011, p. 813.

\bibitem{kupferman2011synthesis}
O.~Kupferman, D.~Sadigh, and S.~A. Seshia, ``Synthesis with clairvoyance,'' in
  \emph{Haifa Verification Conference}.\hskip 1em plus 0.5em minus 0.4em\relax
  Springer, 2011, pp. 5--19.

\bibitem{holtmann2010degrees}
M.~Holtmann, {\L}.~Kaiser, and W.~Thomas, ``Degrees of lookahead in regular
  infinite games,'' in \emph{International Conference on Foundations of
  Software Science and Computational Structures}.\hskip 1em plus 0.5em minus
  0.4em\relax Springer, 2010, pp. 252--266.

\bibitem{zimmermann2017finite}
M.~Zimmermann, ``Finite-state strategies in delay games,'' \emph{arXiv preprint
  arXiv:1709.03539}, 2017.

\bibitem{athanasopoulos2016safety}
N.~Athanasopoulos, K.~Smpoukis, and R.~M. Jungers, ``Safety and invariance for
  constrained switching systems,'' in \emph{2016 IEEE 55th Conference on
  Decision and Control (CDC)}.\hskip 1em plus 0.5em minus 0.4em\relax IEEE,
  2016, pp. 6362--6367.

\bibitem{nilsson2012temporal}
P.~Nilsson, N.~{\"O}zay, U.~Topcu, and R.~M. Murray, ``Temporal logic control
  of switched affine systems with an application in fuel balancing,'' in
  \emph{2012 American Control Conference (ACC)}.\hskip 1em plus 0.5em minus
  0.4em\relax IEEE, 2012, pp. 5302--5309.

\bibitem{bertsekas1972infinite}
D.~Bertsekas, ``Infinite time reachability of state-space regions by using
  feedback control,'' \emph{IEEE Transactions on Automatic Control}, vol.~17,
  no.~5, pp. 604--613, 1972.

\bibitem{de2004computation}
E.~De~Santis, M.~D. Di~Benedetto, and L.~Berardi, ``Computation of maximal safe
  sets for switching systems,'' \emph{IEEE Transactions on Automatic Control},
  vol.~49, no.~2, pp. 184--195, 2004.

\bibitem{rungger2017computing}
M.~Rungger and P.~Tabuada, ``Computing robust controlled invariant sets of
  linear systems,'' \emph{IEEE Transactions on Automatic Control}, vol.~62,
  no.~7, pp. 3665--3670, 2017.

\bibitem{smith}
S.~W. Smith, P.~Nilsson, and N.~Ozay, ``Interdependence quantification for
  compositional control synthesis with an application in vehicle safety
  systems,'' in \emph{Decision and Control (CDC), 2016 IEEE 55th Conference
  on}.\hskip 1em plus 0.5em minus 0.4em\relax IEEE, 2016, pp. 5700--5707.

\bibitem{MPT3}
M.~Herceg, M.~Kvasnica, C.~Jones, and M.~Morari, ``{Multi-Parametric Toolbox
  3.0},'' in \emph{Proc.~of the European Control Conference}, Z\"urich,
  Switzerland, July 17--19 2013, pp. 502--510,
  \url{http://control.ee.ethz.ch/~mpt}.

\bibitem{michiganroad}
\emph{Road Design Manual}.\hskip 1em plus 0.5em minus 0.4em\relax Michigan
  Department of Transportation.

\end{thebibliography}
